\documentclass[12pt]{amsart}
\usepackage{xypic}
\usepackage{amsmath, amssymb}

\DeclareMathOperator{\SO}{SO}

\newcommand{\calR}{\mathcal R}

\newcommand{\bbR}{\mathbb R}

\newcommand{\cF}{\mathcal F}

\newcommand{\calD}{\mathcal D}

\newtheorem{thm}{Theorem}[section]

\newtheorem{lemma}[thm]{Lemma}

\theoremstyle{definition}
\newtheorem{definition}[thm]{Definition}

\newcounter{my_enumerate_counter}
\newcommand{\pushcounter}{\setcounter{my_enumerate_counter}{\value{enumi}}}
\newcommand{\popcounter}{\setcounter{enumi}{\value{my_enumerate_counter}}}

\title{Independence of the existence of Pitowsky  spin models}

\author{Ilijas Farah}
\address{Department of Mathematics and Statistics\\
York University\\
4700 Keele Street\\
North York, Ontario\\ Canada, M3J 1P3\\
and Matematicki Institut, Kneza Mihaila 34, Belgrade, Serbia}

\urladdr{http://www.math.yorku.ca/$\sim$ifarah}
\email{ifarah@yorku.ca}
\author{Menachem Magidor} \address{Einstein Institute of Mathematics\\ Hebrew University of Jerusalem\\ Jerusalem, 91904\\ Israel} \email{mensara@savion.huji.ac.il}
\date{\today.}

\begin{document}

\maketitle

In the last several decades the study of the foundations of Mathematics is dominated by the impact of the independence phenomena in Set Theory. In many fields of Mathematics like analysis, topology, algebra etc. basic problems were shown to be independent of the accepted axiom system for Set Theory, known as Zermello-Frankel Set Theory or ZFC. The fact that undecided problems exist in any recursively stated axiom system which has some minimal strength was well known since G\"odel's famous incompleteness theorem of 1931. But the undecided statement produced by G\"odel's proof was considered to be somewhat esoteric and unrelated to main stream of the subject. The new aspect that came up after Cohen independence proofs of 1963 is that mathematical problems that were considered to be central to the particular discipline were shown to be undecided. So the very notion of mathematical truth was shaken. The possibility of having a multitude of mathematical universes with different properties became a definite option for the foundation of Mathematics.

To what extent these developments are relevant to Physics? When a physical theory is stated in terms of mathematical concepts like real numbers, Hilbert spaces, manifolds etc. it implicitly adapts all the mathematical facts which are accepted by the Mathematicians to be valid for these concepts. If the mathematical ``truths'' may depend of the foundation of Set Theory then it is possible, at least in principle, that whether a given physical theory implies a particular physically meaningful statement may depend on the foundational framework in which the implicitly assumed Mathematics is embedded. 

This may seem far fetched and it is very likely that physically consequences of a physical theory will never depend on the set theoretical foundation of the mathematical reasoning that accompanied the theory but the point of this paper is that this is still a definite possibility. Let us admit from the outset that we do not have an example of a physically meaningful statement that its truth depends on the set theoretical foundation.  We shall instead demonstrate that there are independent statements whose consequences have the same flavour as some statements about the foundation of Quantum Mechanics.

 In \cite{Pitowsky:Deterministic}, \cite{Pitowsky:Resolution}
   I. Pitowsky used Continuum Hypothesis   to construct
hidden variable models for spin-$1/2$ and spin-1 particles in quantum mechanics. 
His functions are not measurable and are therefore not directly subject to 
 Bell-type no-go theorems (\cite{Bell:EPR}, \cite{CHSH}, also cf.  the next paragraph). 
 Also, under the same assumption Pitowsky constructed a function that almost violates the no-go theorem of  
  Kochen and Specker~(\cite{KocSpe:Problem}). This theorem states that one cannot assign the values 0 and 1 to points on the unit sphere in three-dimensional space,~$S^2$, 
  so that the sum of values for   any triple of mutually orthogonal points is~2.   
    Pitowsky constructed a function such that for every vector $x\in S^2$ there are at most countably many exceptions (see Definition~\ref{D.KSP}).  
 We prove  that no such function exists in some model of   the Zermelo--Fraenkel 
set theory with the Axiom of Choice, ZFC 
 (Theorem~\ref{T1}), confirming a conjecture of Pitowsky. While this independence 
 result probably does not have physical interpretation, it gives some weight to the 
 conjecture that one could decide between different set theories on the basis on their scientific consequences. See \cite[\S 5]{Mag:Some} for a discussion of this viewpoint. 
 
 Before proceeding, we should explain why this example still falls short of 
 having direct physical relevance. Pitowsky's function $s_1$ provides correct probabilities for the 
 spin values as 
 predicted by quantum mechanics. However, a complete model for hidden variables
 would require having an infinite sequence of functions $s_1$, $s_2$, $s_3$,\dots
 such that the values provided by $s_n$ model the $n$-th repetition of the  experiment. 
 If $s_n=s_1\circ \rho_n$ for a randomly chosen sequence of
 rotations, then Pitowsky's function gives the correct probabilities for values of spin on many points
 on the sphere (\cite{Pitowsky:Responds}). 
   However, a Bell-type argument shows that 
 it is not possible to correctly model the infinite sequence of 
 values of measurements on a carefully chosen quadruple of 
 points, and the measurability of functions $s_n$ is not required
 to prove this fact~(\cite{mermin}). 

We prove two results. The first one is that 
if there exists a $\sigma$-additive extension of the Lebesgue 
measure to the power-set of the reals then Pitowsky models do not  exist. 
The second is that the same conclusion holds in the model of ZFC known as the
 `random real model.'  
The proofs use results of  H. Friedman and D.H. Fremlin, respectively,    
to show that  under these assumptions every Pitowsky function has to 
be Borel-measurable. 

The first proof has a more elegant assumption and it is easier to understand. 
However, its assumption, also known as `the continuum is a real-valued measurable 
cardinal' is a large cardinal axiom (\cite{Kana:Book}).   By a result of Solovay 
this is equiconsistent with the existence of a measurable cardinal and its consistency 
strength is therefore a bit beyond that of ZFC.   
Our second result shows that the assertion `there are no Pitowsky
functions' is relatively consistent with ZFC. 

\subsection*{Acknowledgments}
The results of this note were conceived in December 2011 during the first author's visit to the Hebrew University of 
Jerusalem. He would like to thank the second author for making this visit most enjoyable and 
productive.
Our results were presented at the von Neumann conference  in M\"unster, May 2012 (I.F.)  
and in the ``Exploring the Frontiers of Incompleteness'' series  at the Harvard University, April 2012 
(M.M.). We would like to thank the organizers of these events. 
I.F. would also like to thank Detlev Buchholz and Marek  Bo\v zejko for illuminating feedback.

\section{Pitowsky functions}
\label{S.Def}
For  $a\in S^2$ and  $0\leq \theta\leq \pi$ let $c(a,\theta)$ be the set of all 
$x\in S^2$ such that  the the angle between $a$ and $x$ is equal to $\theta$ and let 
$\mu_{a,\theta}$ be the normalized Lebesgue measure on $c(a,\theta)$. We shall omit the subscript $a,\theta$ whenever it is clear from the context.

\begin{definition} 
A \emph{Pitowsky spin function}
is a function $F$ from  
  the 2-dimensional sphere~$S^2$ into $[-1,1]$ satisfying 
the following for all $a\in S^2$:
\begin{enumerate}
\item\label{I.P1}  $F(a)=-F(-a)$, and  
\item\label{I.P2}  For  all $\theta\in [0,\pi]$ we have 
$
\int_{c(a,\theta)} F(x) \, d\mu(x) =F(a)\cos(\theta)$. 
\pushcounter
\end{enumerate}
\end{definition}

For  $a\in S^2$ denote the great circle $c(a,\pi/2)$ by $C(a)$.

\begin{definition} \label{D.KSP} 
A \emph{KSP function} is a function $F$ from  $S^2$ into $\{0,1\}$   such that
for all $a\in S^2$ we have 
\begin{enumerate}
  \item $F(-a)=F(a)$, and 
  \item $F\upharpoonright C(a)$ is a measurable function with
   respect to $\mu_{a,\pi/2}$, and 
  \item 
  the set $$\{b\in C(a)|\forall c\in C(a)\cap C(b)F(a)+F(b)+F(c)\neq 2\}$$ is of $\mu_{a,\pi/2}$-measure zero.
\end{enumerate}
\end{definition}
Special cases of Pitowsky spin  functions, 
\emph{spin-$1/2$ functions}  
and 
\emph{spin-$1$ functions}, were constructed by Pitowsky in \cite[\S II]{Pitowsky:Deterministic}
and  \cite[\S IV]{Pitowsky:Deterministic}, respectively, while   
KSP functions were constructed in \cite{Pitowsky:Quantum}. 
All of these constructions use  Continuum Hypothesis
or weaker Martin's Axiom. 

In the following by `random real model' we mean a model obtained by forcing with the homogeneous 
measure algebra
of Maharam character~$(2^{\aleph_0})^+$ (i.e., `adding $(2^{\aleph_0})^+$ random reals') over a model of ZFC. 

\begin{thm} \label{T1}  Assume that the continuum is real-valued measurable
or that the universe is the random real model. 
Then there is no Pitowsky spin  function and there is no  KSP    function. 
\end{thm}

\begin{proof} Assume $F$ is a Pitowsky spin  function. By Lemma~\ref{L.rvm.P} and 
Lemma~\ref{L.P.1}, 
$F$ is Borel-measurable. It is a  well-known consequence
of  the CHSH inequality (\cite{CHSH}) that 
there  are no measurable spin functions. 

Now assume $F$ is a KSP function. By Lemma~\ref{L.rvm.KSP} 
 and Lemma~\ref{L.KSP.1}, 
it is Borel-measurable. From Lemma~\ref{L.no.KSP.1} and the fact that for 
a measurable $F$ the value of $I_{a,b,c}(F)$ 
does not depend on the choice of $a,b$ and~$c$ we can conclude that there  
 are no measurable   KSP functions.  
\end{proof}

J. Shipman (\cite[p. 480]{Shipman}) announced that the conclusion of Theorem~\ref{T1} is 
relatively consitent with ZFC. The proof was not included and  
we  were not able to reconstruct his arguments neither directly nor after corresponding with him. 
In personal communication David Fremlin sketched a clever Fubini-type argument that 
provides an alternative proof of Theorem~\ref{T1}. 

\section{$\cF$-measurability} 

\subsection{Parameterizations} A function $f\colon S^2\to \bbR^2$ is a
\emph{parameterization} if it is injective and measure-class preserving. 
To a pair of points $a$ and $b\neq \pm a$ in $S^2$ we associate
 three parameterizations. 
Let $L(a,b)$ be the half great circle starting from $a$ and going through $b$. For $\theta\in [0,2\pi]$ let $L(a,b,\theta)$ be the half great circle one gets by rotating $L(a,b)$ by the angle $\theta$ clockwise around $a$.

With $a$ and $b$ fixed, to a point $x\in S^2$ associate the following
angles 
\begin{align*} 
\phi^a(x)&=\angle (x,a)\\
\theta^{ab}(x)&\text{ is $\theta\in [0,2\pi)$ such that $x\in L(a,b,\theta)$}
\end{align*} 
Then the following 
are parameterizations with the range equal to $[0,2\pi)\times (0,\pi)$ (plus two points). 
\begin{align*}
p^{ab}(x)&=( \theta^{ab}(x),\phi^a(x))\\
q^{ab}(x)&=(\theta^{ab}(x) ,\theta^{ba}(x)\textrm{ mod }  \pi). 
\end{align*}
Note that $\theta^{ba}(x)$ is taken modulo $\pi$ in order to assure the uniqueness of $q^{ab}(x)$. 

\subsection{$\cF$-measurability} Let $\cF$ be a family of parameterizations. 
A function $F\colon S^2\to \bbR$ is \emph{$\cF$-measurable} 
if for every $p\in \cF$ and all $s$ and $t$ in $\bbR$ all sections
\[
t\mapsto F(p^{-1}(s,t))\text{ and } 
s\mapsto F(p^{-1}(s,t)
\]
are measurable. 
Define
\[
\cF_P=\{p^{ab}, q^{ab}: a\in S^2, b\in S^2, a\neq \pm b\}
\]
and 
\[
\cF_{KSP}=\{q^{ab}: a\in S^2, b\in S^2, a\neq \pm b\}. 
\]
 A function is \emph{P-measurable} if it is $\cF_P$-measurable and 
it is \emph{KSP-mea\-su\-ra\-ble} if it is $\cF_{KSP}$-measurable. Pitowsky spin functions 
are P-measurable while  KSP functions are KSP-measurable.

\subsection{Transitivity}
Parameterizations $p$ and $q$  \emph{have the common first coordinate} 
if the first coordinates of $p(x)$ and $q(x)$ are equal for all $x$. 

\begin{lemma} \label{L2} 
Assume $p$ and $q$ have the common first  coordinate. 
Then for every $X\subseteq S^2$ for almost all~$s$ 
 the $s$-section $Y_s=\{t: (s,t)\in p[X]\}$ is null 
if and only if the $s$-section $Z_s=\{t: (s,t)\in q[X]\}$ is null. 
\end{lemma} 

\begin{proof} 
By the assumption, $q\circ p^{-1}$ is a  measure-class preserving injection from 
 $A=p[S^2]$ onto  $B=q[S^2]$. 
Since $p$ and $q$ have the common first coordinate $q\circ p^{-1}$  maps 
$A\cap (\{s\}\times \bbR)$ 
onto $B\cap (\{s\}\times \bbR)$ for all~$s$. 
Therefore the restriction of $q\circ p^{-1}$ to $A\cap (\{s\}\times \bbR)$ is 
measure-class preserving for almost all $s$. For such $s$ set $Y_s$ is 
null iff $Z_s$ is null. 
\end{proof}

For a parameterization $p$ let $p^T$ denote the parameterization obtained by composing 
$p$ with the flip of coordinates, $(s,t)\mapsto (t,s)$. 
Two parameterizations $p$ and $q$ \emph{have a common coordinate} if 
at least one of $p$ and $p^T$ has the common first coordinate with at least 
one of $q$ and $q^T$. 
A family of parameterizations $\cF$ is \emph{transitive} if the graph whose vertices 
are elements of $\cF$ and edges are pairs $\{p,q\}$ that have a common coordinate
is connected. Note that $\cF_P$ is transitive, but  $\cF_{KSP}$  is not. 
To remedy this, let $q^{ab\gamma}=(\theta^{ab}, \theta^{ba}+\gamma\textrm{ mod } \pi)$. 
Then   
\[
\cF_{KSP}^+=\{q^{ab\gamma}: a\in S^2, b\in S^2, a\neq \pm b, 0\leq  \gamma\leq \pi\}
\]
is transitive and every KSP-measurable function is  $\cF_{KSP}^+$-measurable. 

\subsection{Approximations}
For  $A\subseteq \bbR^2$ and $s\in \bbR$ write 
$A_s=\{t: (s,t)\in A\}$ and $A^s=\{t: (t,s)\in A\}$. 
For $X\subseteq \bbR$ we say that 
$A$ is \emph{$X$-full} if for all $s\in X$ both  $A_s$ and $A^s$ have  full measure.

From now on we assume $\calD$ is a filter on $\bbR$ such that every $X\in \calD$ 
has full outer measure.  
For  a parameterization $q$ and $F\colon S^2\to \bbR$ and $X\subseteq \bbR$
 we say $F$ is \emph{$q$-$X$-approximated} if there is a Borel function 
$G$ such that the $q$-image of $\Delta_{F,G}=\{x: F(x)= G(x)\}$ 
is $X$-full for some~$X$.
If $F$ is $q$-$X$-approximated for some $X\in \calD$ we 
say that it is \emph{$q$-$\calD$-approximated}. 
  It is \emph{$\cF$-$\calD$-approximated} if there is a Borel function $G$ 
which $q$-$\calD$-approximates $F$ for all $q\in \cF$.

\begin{lemma} \label{L.approx} Assume $F$ is $q$-$\calD$-approximated for all $q\in \cF$ 
and $\cF$ is transitive. Then $F$ is $\cF$-approximated. 
\end{lemma} 

\begin{proof} For each $q\in \cF$ fix $X_q\in \calD$ 
 and a  Borel $G_q$ which 
$q$-$X_q$-approximates $F$. 
  Fix  $p$ and $q$ in $\cF$ that have common first coordinate. 
 Let  $A=\{x: G_p(x)\neq G_q(x)\}$. 
 By Lemma~\ref{L2}, $q\circ p^{-1}(X)$ has almost all of its $s$-sections 
 null.  Therefore 
 the measurable set 
 $q[A]$ is covered by the union of two sets each of which has all of its $s$-sections null
 for every $s\in X_q\cap X_p$. Since $X_q\cap X_p$ has full outer measure and
 $q[A]$ is measurable,   
  $q[A]$ is null.  Hence $A$ is null as well. 
 \end{proof}

\section{Fremlin's and Friedman's Fubini-type theorems}

Theorem~\ref{T0} is  an immediate consequence of results of  
D.H. Fremlin and H. Friedman (\cite[Theorem~6K(b)]{Fr:RVM}, \cite{Fri:Consistent}, see also \cite{Shipman}), or rather their proofs, and a negligible
strengthening of the latter result. 
While there is a well-developed machinery for showing that `continuum is real-valued measurable' 
implies many statements true in random real model (see~\cite{Fr:RVM}) we do not have a 
uniform proof of these statements.

\begin{thm} \label{T0} Assume that the continuum is real-valued measurable
or that the universe is the random real model. 
Then there exists a filter~$\calD$ of subsets of $\bbR$ such that 
\begin{enumerate}
\item Every set in $\calD$ has full outer measure, and 
\item For every $f\colon \bbR^2\to \bbR$ such that all of its sections 
$f_s$ and $f^s$ are measurable there exists 
 a Borel 
function $g\colon \bbR^2\to \bbR$ such that 
$\{(s,t): g(s,t)=f(s,t)\}$
is $X$-full for some $X\in \calD$. 
\end{enumerate}
\end{thm}

\begin{proof} Fix $f$. If  the continuum is real-valued measurable then 
let $\calD=\{X: X$ is a Borel set of full measure$\}$. 
By \cite[Theorem~6K(b)]{Fr:RVM}. 
  there exists 
 a measurable function $h\colon \bbR^2\to \bbR$
which $\bbR$-approximates $f$.  
By changing $h$ on a null set we can assume it is Borel-measurable
and~$X$-approximates $f$ for some Borel set of full measure. 

Now assume we are in a model obtained by adding $(2^{\aleph_0})^+$ random reals 
to a model of ZFC. For simplicity of notation we assume that the Continuum Hypothesis 
holds in the latter model. 
Let $\calR$ denote the set  of random reals. 
It was proved by Friedman (\cite[Lemma~15 and Theorem~2]{Fri:Consistent})
that some Borel function $g$ $X$-approximates $f$ for a set $X$ of full outer measure. 
A closer look at the proof of Lemma~15 
reveals that $X\subseteq \calR$ is the set of all random reals added over an intermediate 
model in which CH holds. Such set always has full outer measure. 
Therefore $\calD=\{X\subseteq \calR: |\calR\setminus X|\leq \aleph_1\}$ is a filter 
as required. 
\end{proof}

\subsection{Pitowsky spin functions}

Throughout this section we assume conclusion of Theorem~\ref{T0} holds for 
a filter $\calD$ of full outer measure subsets of $\bbR$.  
The following is immediate from Theorem~\ref{T0} and Lemma~\ref{L.approx}. 

\begin{lemma} \label{L.rvm.P} 
Every bounded $P$-measurable function is P-$\calD$-ap\-prox\-i\-ma\-ted. \qed
\end{lemma}

\begin{lemma} \label{L.P.1} If  a Pitowsky spin function $F$ is P-$\calD$-approximated by  
a Borel function $G$  
 then it is Borel. 
\end{lemma} 

\begin{proof}For every $a\in S^2$ there is $X\in \calD$ such that 
for all $\theta$ in $X\cap [0,\pi/2)$
we have that 
\[
F(a)=\frac 1{\cos(\theta)} \int_{c(a,\theta)} F \, d\mu=
\frac 1{\cos(\theta)} \int_{c(a,\theta)} G\, d\mu. 
\]
Since the function $\theta\mapsto \int_{c(a,\theta)} G\, d\mu$
is measurable and constant on a set of full outer measure, it is almost everywhere equal to 
$F(a)$. Therefore  
 $F(a)=\frac 2\pi \int_0^{\pi/2} \frac 1{\cos(\theta)} \int_{c(a,\theta)} G\, d\mu$
for all $a$, and it only remains to check that the function on the right hand side of the 
equality is Borel. 
Recall that $P(S^2)$ denotes the compact metric 
space of probability Borel measures on $S^2$ (see \cite[\S 17]{Ke:Classical}).
Function $(a,\theta)\mapsto c(a,\theta)$ from $S^2\times [0,\pi)$ 
to $K(S^2)$ (the space of compact subsets of $S^2$) is continuous, 
and by the Portmanteau Theorem (\cite[17.20]{Ke:Classical}) so is the function 
$(a,\theta)\mapsto \frac 1{\sin(\theta)} \mu_{a,\theta}$ from $S^2\times [0,\pi)$ into 
$P(S^2)$. Finally, the map $\mu\mapsto \int G\, d\mu$ is a Borel map from $P(S^2)$ into 
$\bbR$ by \cite[17.24]{Ke:Classical}, and the conclusion follows. 
\end{proof}

\subsection{KSP functions}

Throughout this section we assume the conclusion of Theorem~\ref{T0} holds for 
a filter $\calD$ of full outer measure subsets of~$\bbR$.  
Since $\cF_{KSP}$ is  transitive, a proof identical to the proof of 
Lemma~\ref{L.rvm.P}  gives the following lemma. 

\begin{lemma} \label{L.rvm.KSP} 
Every bounded KSP-measurable function is KS-$\calD$-ap\-prox\-i\-ma\-ted. \qed
\end{lemma}

\begin{lemma}\label{L.KSP.1}If a KSP function $F$ is KS-$\calD$-ap\-prox\-i\-ma\-ted by 
a Borel function $G$ then it is 
Borel. 
\end{lemma} 

\begin{proof} 
We have $F(x)=1-2\int_{c(x,\pi/2)} F\, d\mu$ for every $x\in S^2$. Fix $a\in S^2$. Then the set of 
 $x\in c(a,\pi/2)$ such that $\int_{c(x,\pi/2)} G\, d\mu=\int_{c(x,\pi/2)} F\, d\mu$ has full outer measure
 in $c(a,\pi/2)$. 
 The function of $x\in c(a,\pi/2)$ on the right-hand side is measurable (by \cite[17.24]{Ke:Classical}) 
 and agrees with $F$ on a set of full outer measure
 in $c(a,\pi/2)$.   Since $F$ is Borel, these two functions agree almost everywhere and
 \begin{align*}
F(a)&=1-2\int_{c(a,\pi/2)} (1-2\int_{c(x,\pi/2)}  G(y)\, d\mu_{x,\pi/2}(y)) \, d\mu_{a,\pi/2}(x)
\end{align*} 
and it only remains to check that the function of $a$ on the right hand side is Borel-measurable. 
Like in the proof of Lemma~\ref{L.P.1}, this follows from  $a\mapsto \mu_{a,\pi/2}$ being Borel 
and applying  \cite[17.24]{Ke:Classical} twice. 
\end{proof} 

Fix two members of $S^2$ $a,b$ which are  orthogonal. We represent each rotation $\alpha$ in $\SO(3)$ as a product $\alpha_3^{a,b}(r,s,t)\cdot\alpha_2^{a,b}(r,s)\cdot\alpha_1^{a,b}(r)$ where $\alpha_1^{a,b}(r)$ is a rotation around $a$ by an angle $r$, $\alpha_2^{a,b}(r,s)$ is a rotation around $\alpha_1^{a,b}(r)(b)$ by an angle $s$  and $\alpha_3$ is a rotation around $\alpha_2^{a,b}(r,s)(\alpha_1^{a,b}(r)(b))$ by an angle $t$. (We are going to omit the superscript $a,b$ if they are understood from the context.)

For any function $F$ on $S^2$ and $a,b,c\in S^2$ which are mutually orthogonal define 
\[
I_{a,b,c}(F)=\int_0^{2\pi}\int_0^{2\pi}\int_0^{2\pi}F(\alpha_3^{a,b}(r,s,t)\cdot\alpha_2^{a,b}(r,s)\cdot\alpha_1^{a,b}(r)(c))\,\mathrm{d}t\,\mathrm{d}s\,\mathrm{d}r
\]
where the measure on on the interval $[0,2\pi]$ is normalized to be 1 and $I_{a,b,c}(F)$ is defined only if all the relevant integrals are defined. Note that if $G$ is a measurable function on $S^2$ then $I_{a,b,c}(G)=\int_{\SO(3)}G(\alpha(c))\,\mathrm{d}\alpha$ where we take the  right invariant Haar measure on $\SO(3)$. In this case $I_{a,b.c}(G)$ does not depend on the choice of $a,b,c$. 

\begin{lemma} \label{L.no.KSP.1} Assume  $F$ is a KSP function on $S^2$. 
 Then for every mutually orthogonal triple $a,b,c$ $I_{a,b,c}(F)$ is defined and is equal to
$(5-F(a))/8$
\end{lemma}

\begin{proof}
By $F$ being a KSP function it is easily seen that for every $d\in S^2$ we have 
 $\int_{c(d,\pi/2)} F\, d\mu=1-\frac{F(d)}{2}$. Fix $r,s\in[0,2\pi]$ . The inner integral in the definition of $I_{a,b,c}(F)$ is the integral of $F$ over $C(\alpha_2(r,s)(a))$. Hence it is equal to $1-\frac{F(\alpha_2(r,s)(a))}{2}$. The second integral (according to~$s$) is the integral of the function $H(x)=1-\frac{F(x))}{2}$ on the great circle $C(\alpha_1(r)(b))$ so it is equal to $(3+F(\alpha_1(r)(b))/4)$. The last integral is the integral of the function $K(x)=(3+F(x))/4)$ on the great circle  $C(a)$ so it is equal to $(5-F(a))/8$.
\end{proof}

\providecommand{\bysame}{\leavevmode\hbox to3em{\hrulefill}\thinspace}
\providecommand{\MR}{\relax\ifhmode\unskip\space\fi MR }
\providecommand{\MRhref}[2]{%
  \href{http://www.ams.org/mathscinet-getitem?mr=#1}{#2}
}
\providecommand{\href}[2]{#2}

\end{document}